\documentclass[draftclsnofoot,onecolumn,12pt]{IEEEtran}
\usepackage{amsthm,amssymb,graphicx,graphicx,multirow,amsmath,color,cite}
\usepackage[latin1]{inputenc}

\def\E{\mathbb{E}}

\def\P{\mathbb{P}}
\def\ie{{\em i.e.}}
\def\eg{{\em e.g.}}

\def\R{\mathbb{R}}

\def\eta{\lambda}

\def\i{\mathbf{1}}
\def\d{\mathrm{d}}
\def\pd{\varrho}
\def\M{\mathtt{M}}
\def\K{\mathtt{K}}
\def\p{\mathtt{P_s}}

\def\x{\mathtt{x}}
\def\b{\mathtt{b}}
\def\G{\mathtt{G}}
\def\W{\mathtt{N_o}}
\def\v{\mathtt{\nu}}
\def\y{\mathtt{y}}

\def\I{\mathtt{I}}
\def\pp{\mathtt{P}}
\def\ps{\mathtt{W_{sd}}}
\def\F{\mathtt{F}}
\def\H{\mathtt{H}}

\def\PP{\mathbb{P}^{!o}}
\def\EP{\mathbb{E}^{!o}}
\def\T{\theta}
\def\sinr{\mathtt{SINR}}

\def\calN{\mathcal{N}}
\def\FE{\mathcal{F}}

\def\l{\ell}

\newcommand{\h}[1]{\ensuremath{\mathtt{h}_{#1}}}

\newcommand{\err}[1]{\ensuremath{\operatorname{Err}(\eta,#1)}}
\newcommand{\FD}[1]{\ensuremath{|\mathcal{F}_{#1}|}}
\newtheorem{lemma}{Lemma}{}
  \newtheorem{thm}{Theorem}
  \newtheorem{theorem}{Theorem}
  
  \newtheorem{cor}[thm]{Corollary}
  
  \newtheorem{definition}[thm]{Definition}

\begin{document}

\title{Series Expansion for Interference in Wireless Networks}
\author{Radha Krishna Ganti, Fran\c{c}ois Baccelli and Jeffrey G. Andrews
\thanks{F. Baccelli is with  Ecole Normale Sup\'erieure(ENS) and INRIA in Paris, France, J. G. Andrews and R. K. Ganti are with the University of Texas at Austin.  The contact author is R. K. Ganti, rganti@austin.utexas.edu. Date revised: \today}}

\maketitle

 \begin{abstract}
The spatial correlations in transmitter node locations introduced by common multiple access protocols makes the analysis
of interference, outage, and other related metrics in a  wireless network extremely
difficult. Most  works therefore assume that nodes are
distributed either as a Poisson point process (PPP) or a grid, and utilize the independence
properties of the PPP (or the regular structure of the grid) to analyze interference, outage and other metrics. But,
the independence of node locations makes the PPP a dubious model for nontrivial MACs which intentionally introduce correlations, e.g. spatial separation, while the grid is too  idealized to model real networks.  In this paper, we introduce a new technique based on the factorial
moment expansion of functionals of point processes to analyze functions of
interference,  in particular outage probability. We provide
a  Taylor-series type expansion of functions of interference, wherein
increasing the number of terms in the series provides a better approximation
 at the cost of increased complexity of computation. 
Various examples  illustrate  how this new approach can be used to find outage probability in both Poisson and non-Poisson wireless networks.
\end{abstract}

\section{Introduction}

The spatial distribution of transmitters is critical in determining the mutual interference and hence the performance of a wireless network. A common and ubiquitous model for the node locations is the Poisson point process, where the wireless node locations are  { independent} of each other. While this model offers analytical tractability, it is  insufficient in many cases as it precludes intelligent scheduling, which leads to correlations between node locations.
  Another classical model for the node locations is the grid, but it  is inflexible and too regular to model most  wireless networks. 
  To optimize the spatial packing of transmissions in a wireless network, and obtain the maximum throughput, it is essential that smart MAC protocols are designed and used.
   In order to do so,  new mathematical tools are required to analyze and understand the performance of general spatial networks.

While interference and related metrics have been studied in a few non-Poisson
spatial networks, there is no systematic approach or a general technique to
analyze these problems. In this paper, we provide a new technique to obtain an
expansion for wide range of functions of interference for any spatial distribution of nodes.
As in a Taylor series,  a better 
approximation  of the interference functional can be obtained by increasing the number of terms in the series,  with each additional term being increasingly complicated.  As we shall see, this new technique can be used to analyze functionals of interference for a wide spectrum of MAC protocols, spatial distribution of nodes, and radio design choices. While in this paper we focus on treating interference at the receiver as noise, the techniques in this paper can be easily extended to other sophisticated interference mitigation techniques.

Similar to the moments of a scalar valued random variable, the product densities
\cite{stoyan} specify a stationary spatial point process. For a Poisson
point process, the $n$-th order product densities  can be easily obtained. This is not the case for other point processes. 
In practice,  it is  only feasible  to obtain   a few product densities (with
certain reliability) by analyzing the spatial data provided.  Therefore a non-PPP spatial distribution is usually partially specified by a few product
density measures.   
%Using the new series technique introduced in this paper,  an approximation of  an arbitrary function of interference can be obtained using a limited number of product densities. 
The present paper uses factorial moment expansion techniques  \cite{blaszczyszyn1995factorial} to approximate an arbitrary function of interference  using a limited number of product densities.
%This approximation can be improved  without limit to the extend that the underlying  node distribution (in terms of the product densities) is known.  
The technique provided in this paper extends the rich set of results for interference and outage characterization available for the PPP.

\subsection{Background and Related Work}
In a spatial wireless network, the success (or outage) probability of a {\em typical link} is an important metric of performance, and considerable prior work focused on obtaining the success probability for different spatial network models. The success probability of a typical link is equal to 
\begin{equation}\p = \PP\left(\frac{S}{\I+\W} \geq \T \right),
\label{eq:001}
\end{equation}
where $S$ is the signal power of the typical transmitter, $\I$ is the interference, and $\W$ is the noise power. The reduced Palm measure $\PP$ is used since we are interested in the probability of a "typical link", which essentially corresponds to the conditional probability for point processes.  One can interpret the success probability (or other quantities) of a typical link as the success  probability averaged over all the spatial links \cite{verejones}. 

{\em Poisson Networks}: 
The two main reasons for the analytical tractability of a Poisson point process are its easy Palm characterization due to  Slyvniak's theorem \cite{stoyan}, and the knowledge of its probability generating functional (PGFL).
%There are two main reasons for the analytical tractability of a Poisson point process:
%\begin{enumerate}
%\item The reduce Palm measure $\PP$  \cite{verejones,stoyan} has an easy characterization provided by Slivinak's theorem \cite{stoyan}.  In this case, $\PP =\P$ which implies that placing an additional point at the origin would not effect the other points of the process.
%\item The probability generating functional of a stationary PPP of density $\lambda$ is given by $\G[f]=\exp(-\lambda \int_{\R^2}(1-f(x))\d x)$ and results from the independence properties of the PPP.
%\end{enumerate}
%
 
Using the PGFL of the PPP, the Laplace transform of the interference in a PPP network is obtained in \cite{Sousa90, bacelli-aloha, chan2001calculating, hellebrandt1997cip,WinPin09}, and when the path loss model is given by $\|x\|^{-\alpha}$, $\alpha>2$,  the interference distribution is stable  with parameter $2/\alpha$. 
Obtaining the outage probability in  closed form using the distribution of interference is not always possible, but  bounds on the CDF of the interference distribution can be obtained \cite{WebAnd2006d, WebAndJin07,ganti-2007} which lead to bounds on the outage probability. 
 Using the Laplace transform, the outage probability can be obtained in  closed form for  an   exponential family of fading distributions including Rayleigh fading \cite{bacelli-aloha},  Nakagami-m fading  \cite{hunter-2007}, and  $\chi^2$ and other related distributions in  \cite{jindal-09,ali-2010,huang-08,zeidler2009spatial}.
 In \cite{bartek-09}, a novel Fourier transform-based technique is used to analyze the outage probability in a PPP network with arbitrary fading, but it requires evaluation of complex valued integrals.
Using  ALOHA  to schedule nodes in a PPP network also results in a PPP transmitter set since the independence between transmitters is preserved.  Since other MAC protocols typically  induce correlation,  all the above techniques implicitly or explicitly assume an ALOHA MAC  protocol.

{\em Non-Poisson networks}: In a PPP network, two transmitting nodes can be arbitrarily close, something which typically does not happen in real networks  because of  physical constraints and MAC scheduling. For example,  the popular CSMA MAC protocol prevents two nearby nodes from transmitting at the same time. So spatial models that account for repulsion between transmitting  nodes or other post-MAC correlations are required to model intelligent scheduling protocols.  
  In \cite{nguyen-07, baccelli2009stochastic}, the node locations are modeled by a PPP
with a  modified CSMA MAC protocol, resulting in a transmitter set which is modeled by
a modified Mat\'ern hard-core point process. However,  the success probability
is approximated by assuming a non-homogeneous PPP with intensity equal to the
second-order product density of the Mat\'ern hard-core process.  In
\cite{HasAnd07} the CSMA protocol is modeled by excluding interferers from a
guard zone around the typical receiver in a Poisson
model and  \cite{kaynia2008performance,kaynia-joint} analyze outage probability by using a space-time Poisson approximation of the transmit process and the incoming traffic.
 Other common approximation of the transmit node locations in a CSMA protocol are regular lattices  \cite{haenggi-09, silvester, ferrari,mathar1995dci}.

%  Cognitive networks utilize the  require a spatial model with minimum node separation, \cite{} and \cite{} provides an overview of the point processes used to model these networks.

Clustering of nodes might occur due to environmental conditions (office spaces, gathering spots), be intentionally induced by the MAC protocol \cite{yang2007inducing}, or occur in sensor networks to increase the lifetime of the network \cite{younis2004heed}.
 In \cite{ganti-2007}, the conditional probability generating functional of a Poisson clustered process is obtained, from which the outage probability is derived, and in \cite{tresch-10}, the outage analysis was extended to clustered networks which use intra-cluster interference alignment. 
Femtocells \cite{claussen2008overview,chandrasekhar2008femtocell,chandrasekhar2007uplink} are another example of clustered wireless networks.

We observe that while networks that can be modeled by a Poisson point process have been extensively analyzed, this is not the case with networks  which induce spatial correlations among transmit nodes. They are generally analyzed by making simplified assumptions about the spatial model, or by mathematical approximations in the analysis.  Even in the case of Poisson networks, the outage analysis becomes complicated when one moves away from exponential forms of fading.  Essentially, there is no single mathematical tool or technique that is flexible enough to  analyze these general spatial networks.

 \subsection{Contributions and Organization of the Paper} 

    The  main contribution of this paper is a new mathematical tool that provides an arbitrary {\em close} approximation   to any well-behaved functions of interference in a spatial network.  We use this tool to provide a series approximation of the success probability in  {\em any} stationary wireless network.   Compared to a recent technique that can be used to analyze outage probability of any spatial network in the low-interference regime \cite{ricardo-10}, the present result is far more general since \cite{ricardo-10} requires an asymptotically small user density. 

%Mathematically, there is a close analogy between spatial networks and queuing networks, where again the  analysis of $M/M/1$ (corresponds to a one dimensional PPP) queue is the most tractable. 
In this paper
we use a technique first developed to evaluate the derivatives of queue functionals \cite{reiman1989open,Baccelli_diff}, and later extended to obtain series expansion of functionals of Poisson point processes \cite{baccelli1996taylor,zazanis1992analyticity}. 
The series expansion termed the {\em factorial moment} expansion was generalized to one dimensional non-Poisson point processes in \cite{blaszczyszyn1995factorial}, and to higher dimensional point processes in \cite{baszczyszyn1997note}. 
In \cite{baszczyszyn1997note}, the FME results of \cite{blaszczyszyn1995factorial} were extended by using measurable orderings of the points of the process.  
The main tool we use in this paper is from \cite{blaszczyszyn1995factorial} which provides a factorial moment expansion of functionals of point processes.
The sufficient conditions  for the spatial FME to exist \cite{blaszczyszyn1995factorial} are difficult to verify, and an additional contribution of this paper is the simplification of these sufficient conditions for functionals of interference.
%Although in this paper, we consider interference as noise, the techniques introduced in this paper can be extended to analyze complex interference mitigating techniques.

The paper is organized as follows. In Section \ref{sec:sys_model}, we introduce the system model, and define point process measures. In  Section \ref{sec:FME} we introduce expansion kernels and the main theorem which deals with sufficient conditions for the FME of interference functions.
In Section \ref{sec:out}, we provide the FME  for outage probability, and provide various examples to illustrate the use of FME.

\section{System Model}
\label{sec:sys_model}
The transmitters locations are modeled by a simple, stationary, and isotropic point process
\cite{stoyan,verejones}
%\cite{hough2006determinantal,soshnikov2000determinantal} 
 $\Phi$ on the plane $\R^2$ of density $\eta$.  Each transmitter
$x\in \Phi$ is associated with a mark $\h{x}$ that is independent and
identically distributed and 
does not depend on the location $x$. The random variable $\h{x}$ may
represent the transmit
power, or the small-scale fading between the transmitter and some point
on the plane.
%  Since $\h{x}$ are independent of location,
%they do not fundamentally change the structure of the underlying
%spatial
%locations.
The path loss model denoted by $\l(x): \R^2 \rightarrow [0,\infty)$ is
 a non-decreasing function of $\|x\|$, and further $\int_{B(o,\epsilon)^c} \l(x) \d x <\infty$, $\forall \epsilon >0$, where $B(x,r)$ represents a ball of radius $r$ centered around $x$.
%In this paper we concentrate on the following two path loss functions.
%\begin{enumerate}
%  \item Non singular path loss model: $\l(x)=(1+\|x\|^{\alpha})^{-1}$.
%  \item Singular path loss model: $\l(x)=\|x\|^{-\alpha}$.
%\end{enumerate}
The interference at $y \in \R^2$ is 
\begin{equation}
\I(y,\Phi)= \sum_{\x \in \Phi}\h{\x}\l(\x-y).
\end{equation}
Also,  let
\[\I_z(y,\Phi)= \sum_{\x \in \Phi\cap B(o,\|z\|)}\h{\x}\l(\x-y),\]
be a restriction  on the interferers  to the ball $B(o,\|z\|)$. 
In most cases, we are interested in the performance of a typical
transmitter and its associated receiver. We condition on the event that
a point of the process $\Phi$ is located at the origin and we consider
the node at the origin as a typical transmitter. For example, the outage
probability of a typical transmitter and its receiver at $r(o)$ is given
by  \[\p = \PP\left(\frac{S}{\W+\I(r(o),\Phi)}>\T \right),\]
where $S$ is the received signal power and $\W$ is the noise power and $\PP$ the Palm measure \cite{stoyan}. The $\sinr$ threshold $\T$ for successful communication depends on the required rate, the receiver structure, and the coding scheme used. 
 In this paper we
shall provide a series expansion of functions of interference, more
precisely expansions of $\EP[\F(\I(y,\Phi))]$ for well-behaved functions
$\F(x)$. For example,
\begin{enumerate}
\item The outage probability of a typical link is equal to $\EP[\G(\T
(\I(r(o),\Phi)+\W))]$, where $\G$ is the CCDF of $S$. Hence in this case
$\F(x)=\G(\T x+\T\W)$.
\item The ergodic capacity of a typical link is equal to
$\EP[\log(1+\frac{S}{\W+\I(r(o),\Phi)})]$, so
\[\F(x)= \E_S\left[\log\left(1+\frac{S}{\W+x}\right)\right].\]
\end{enumerate}
We now introduce a few definitions concerning point processes, and we begin by defining an order among the points of the process.
 \begin{definition}[Measurable order]
 For the point process $\Phi$, define an ordering of the points based on 
the distance from the origin. For a simple  and stationary point process this ordering
is a.s.  unique \cite{verejones} and
for $x, y \in \Phi$, we denote $x\preccurlyeq y$ if $\|x\|\leq \|y\|$. 
\end{definition}
We now provide formal definitions of moment measures and Palm probabilities that are used later in this paper. Let $M$ denote the set of finite sequences on $\R^2$, \ie, the set of sequences $\phi =\{x_i\} \subset \R^2$ such that  $|\phi \cap B| <\infty$ for all bounded Borel $B \subset \R^2$ and $x_i\neq x_j$, $i\neq j$ .  Denote by $\mathcal{B}(M)$ the smallest $\sigma$-algebra on $M$ that makes the maps $\phi \to |\phi \cap B|$ measurable for all Borel  $B\subset \R^2$.   For notation simplicity  we define $\phi(B) \triangleq |\phi \cap B|$, $B\subset \R^2$. A probability measure $\P$ on $(M,\mathcal{B}(M))$ defines a point process. 
Product densities characterize the distribution of a point process, and for many point processes are easy to characterize. 
 \begin{definition}[Product densities]
Let $B_i \subset \R^2$, $1\leq i\leq n$. The $n$-th order factorial moment
measure  of a point process  $(\Phi, \P)$ is defined as
   \begin{align}
	 \M^{(n)}_{\P}(B_1,\dots, B_n) &=
	 \E\sum_{x_1,\dots,x_{n} \in \Phi}^{\mathtt{p.d}}
	 \i(x_1\in B_1,\dots,x_n\in B_{n}),\nonumber
	 \end{align}
	 where $\sum^{\mathtt{p.d.}}$ represents sum over pairwise distinct tuples.
The $n$-th order product density of a  point process is defined in terms of the factorial moment measure by
the following relation.	 
	\begin{align} 
	\M^{(n)}_{\P}(B_1,\dots, B_n) 
	 &=\int_{B_1 \times \dots \times
   B_{n}}\pd^{(n)}(x_1,\dots,x_{n})\d x_1\dots\d x_{n}.
	 \label{eq:prod_den}
   \end{align}
   \end{definition}
So $	 \M^{(n)}_{\P}(B_1,\dots, B_n)$ counts the mean number of $n$-tuples
in the set $B_1\times B_2\times \hdots \times B_n$, with no two components of the $n$-tuple
being the same, and the $n$-th order product density is the Radon-Nikodym derivative
of the product measure with respect to the Lebesgue measure.   For a stationary PPP of density $\eta$,  $\pd^{(n)}(x_1,\dots,x_n) =\eta^{n}$. This follows from the independence of node locations.  Also for any stationary point process,
the $n$-th order product density is a function of $y_1=x_2-x_1$,$\hdots$,$y_{n-1}=x_n-x_1$, \ie,
\[\pd^{(n)}(x_1,\dots,x_n)= \pd^{(n)}(y_1,\dots,y_{n-1}).\]
Intuitively,  the $(n+1)$-th order product density  $\pd^{(n+1)}(x_1,\dots,x_{n})$ of a stationary point process is proportional to the probability of finding points of the process at $o, x_1,\hdots, x_n$.  
In this paper we  assume that the $n$-th order factorial moment measures are $\sigma$-finite,  product densities of all orders exist for  the point processes in consideration. Also, we only focus on stationary point processes.

Palm probabilities are the point process counterparts of conditional probabilities of a real valued random variable ans are defined in terms of  Campbell  measures.  The $n$-th order reduced Campbell measure  of the point process $(\Phi,\P)$ is a measure on $(\R^2)^n \times  M$ defined as
\[C^{(n)}(B,A)= \E\left[\int_B \i(\Phi\setminus\{x_1,\hdots,x_n \} \in A)\d( \M^{(n)}_{\P}(x_1,\hdots,x_n) )\right], \quad A\in \mathcal{B}(M), B \subset (\R^2)^n.\]
Please refer to \cite{blaszczyszyn1995factorial,baszczyszyn1997note} for a detailed description of these measures, notation and their properties.
 \begin{definition}[$n$-fold reduced Palm measure ]
The $n$-fold reduced Palm measure is the Radon-Nikodym derivative of the $n$-th order reduced Campbell's measure with respect to the $n$-th order factorial moment measure evaluated at $(x_1,\hdots,x_n)$. More formally, the $n$-fold reduced Palm measure
 is given by 
\begin{equation}
\P_{x_1,\dots,x_n}^{(n)}(A) =  \frac{\d C^{(n)}(\cdot \times A)}{\d C^{(n)}(\cdot \times M)}(x_1,\hdots,x_n), \quad A\in \mathcal{B}(M).\end{equation}
The dot ``$\cdot$''  in the definition represents the variable on which the Radon-Nikodym derivative is defined.
%where $C^{(n)}$ represents the reduced Campbell measure and $\mathcal{M}$ represents the space of all simple sequences (the complete event) \footnote{The dot ``$\cdot$''  in the definition represents the variable on which the Radon-Nikodym derivative is defined.}. 
 \end{definition}
 Informally, the reduced  $n$-fold Palm measure corresponds to the law of the point process given that it has points at $x_1,\hdots, x_n$, excluding these points. 
For notational simplicity we shall denote the $1$-fold Palm measure at the origin by $\PP$.  
%Informally, the reduced Palm probability $\PP(A)$ corresponds to the probability of the event $A \subset \mathcal{M}$, conditioned on the event that there is a point of the process $\Phi$ at the origin, but not counting it.  
%For example, $\PP(\Phi\cap B(o,x) =\emptyset)$ corresponds to the CCDF of the nearest neighbor distribution.
For a stationary point process of density $\lambda$, the first order reduced Palm measure $\PP$  has a simple representation \cite{stoyan}:
\[\PP(A)=\frac{1}{\lambda |B|}\E\left[\sum_{\x \in \Phi \cap B}\i(\Phi_{\x}\setminus \{\x\} \in A)\right], \quad A \in \mathcal{B}(M),\]
where $B$ is a Borel set, $|B|$ its Lebesgue measure, and $\Phi_\x$ corresponds to the translation $\Phi-\x$.   
  
    In the  next section, we introduce the factorial moment expansion (FME) of functionals of interference and provide sufficient conditions for the FME  of interference functions.

\section{Factorial Moment Expansion}
\label{sec:FME}
The Factorial moment expansion was introduced in \cite{blaszczyszyn1995factorial} for point processes on the line and was later extended to spatial point processes in \cite{baszczyszyn1997note}. FME can be considered as a Taylor series (of an analytic function) kind of expansion for functionals of point processes. Similar to a Taylor series, the average of a functional of a point processes is represented as a finite series and an error term that diminishes as the number of terms increases. In a Taylor series, the terms of the series depend on the derivatives of the function, and in FME these derivatives are replaced by expansion kernels which we define below.
%For the sake of notational convenience fixing a point $x$ of the point process also implies fixing the mark (fading) associated with it. 
\subsection{Expansion Kernels}
 We first introduce some notation.
\begin{enumerate}
\item Let $\F(x):\R^+\rightarrow[0,\infty]$ be a real function.  Hence $\F(\I(y,\Phi))$, is 
a functional from the space of marked point processes to real numbers, more precisely defined as \[ \F(\I(y,\Phi)) \triangleq \F(\sum_{\x \in \Phi}\h{\x}\l(\x-y)).\]
Hence $\F(\I(y,\Phi))$ should be interpreted as a composition of a function $\F(x)$ and the interference functional, rather than $\F$ being a function of $\I$.  Similarly, $\F(\I_z(y,\Phi))$  should be interpreted as  $\F(\I(y,\Phi\cap B(o,\|z\|)))$. 

\item Adding  a new point to $\Phi$ corresponds to adding a tuple $(x, \h{x})$ to the point process. But for notational convenience, we just represent it as $\Phi\cup \{x\}$.  So when a point is added to the process, it implicitly means the corresponding mark (fading) is also added. 
\end{enumerate}
\begin{definition}[Continuity]
\label{def:cont}
The functional $\F(\I(y,\Phi)))$ is continuous at $\infty$ if
\[\lim_{\|z\| \rightarrow \infty}\F(\I_z(y,\Phi))=\F(\I(y,\Phi))),\]
holds  true for any simple and finite point\footnote{A simple and finite point process  is a point process for which $\Phi(B) <\infty$, for $|B|<\infty$, and no two points coincide. } set $\Phi$.
\end{definition}

\begin{definition}[Expansion Kernels]
 Let $\F(x):\R^+ \rightarrow [0,\infty]$ be a real function,  the first order expansion kernel is defined  as
\[\F^{(1)}_{z}(\I(y,\Phi)) =
\F(\I_z(y,\Phi)+\h{z}\l(z-y))-\F(\I_z(y,\Phi)),\]
and the $n$-th order expansion kernel is defined by
\[\F^{(n)}_{z_1,\hdots,z_n}(\I(y,\Phi))=(\hdots(
\F^{(1)}_{z_1})_{z_2}^{(1)}\hdots
)_{z_n}^{(1)}(\I(y,\Phi)).\]
\end{definition}
As mentioned in Remark 2, $\F^{(n)}_{z_1,\hdots,z_n}(\I(y,\Phi))$ actually
means
$\F^{(n)}_{(\h{z_1},z_1),\hdots,(\h{z_n},z_n)}(\I(y,\Phi))$,  \ie,
 the kernel also involves the marks of the
 added points and not only the points. In this sense, even
 when applied to 0 (the point measure $\Phi$ with zero mass), the kernel is random. 
\noindent For example, 
%{\bf THERE IS A PROBLEM HERE CHECK}
 \begin{align*}
   \F^{(2)}_{z_1,z_2}(\I(y,\Phi))=& 
\F(\I_{z_2}(y,\Phi))\\
&-\F(\I_{z_2}(y,\Phi)+\h{z_2}\l(z_2-y))-\F(\I_{z_2}(y,\Phi)+\h{z_1}\l(z_1-y))\\
&+\F(\I_{z_2}(y,\Phi)+\h{z_1}\l(z_1-y)+\h{z_2}\l(z_2-y)),\quad
z_2\preccurlyeq z_1.\end{align*}
This expansion kernels can be written in a compact form\cite{blaszczyszyn1995factorial}:
\begin{equation}
\F^{(n)}_{z_1,\hdots,z_n}(\I(y,\Phi))=
\left\{\begin{array}{ll}
\sum_{j=0}^n(-1)^{n-j} \sum_{\Pi \in \left\{\binom{n}{j}\right\}}
 \F(\I_{z_n}(y,\Phi)+\sum_{i\in \Pi}\h{z_i}\l(z_i-y)),&
 z_n\preccurlyeq,\hdots,\preccurlyeq z_1,\\
0,&\text{otherwise},
\end{array}\right.
\end{equation}
where $\left\{\binom{n}{j}\right\}$  denotes the collection of all
cardinality  $j$ subsets of $\{1,\hdots,n\}$. 
 Furthermore the null-kernel is defined as,
\begin{equation}
\F^{(n)}_{z_1,\hdots,z_n}(0)=
\left\{\begin{array}{ll}
\sum_{j=0}^n(-1)^{n-j} \sum_{\Pi \in \left\{\binom{n}{j}\right\}}
 \F(\sum_{i\in \Pi}\h{z_i}\l(z_i-y)),&
 z_n\preccurlyeq,\hdots,\preccurlyeq z_1,\\
0,&\text{otherwise},
\end{array}\right.
\end{equation}
Observe that $\F^{(n)}_{z_1,\hdots,z_n}(0)$  is a random variable because of the added marks.

    \subsection{Factorial Moment Expansion of Interference Functionals}
The following theorem deals with the FME of interference functionals, and extends Theorem 3.2 in \cite{blaszczyszyn1995factorial} for the reduced Palm measure and interference. Note that we state the theorem for a function of interference, although the main result holds for  any functional of a point process.  Let $\E_{\h{x_1},\dots,\h{x_i}}$ denote the expectation with respect to the random variables ${\h{x_1},\dots,\h{x_i}}$.
 \begin{theorem}[FME]
 \label{thm:main}
   Let $\F(\I(y,\Phi))$ be be such that the functional  is continuous at  infinity, and such that
   \begin{equation}
	 \int_{\R^{2i}}\int_M
\left|\E_{\h{x_1},\dots,\h{x_i}}\left[\F^{(i)}_{x_1,\dots,x_i}(\I(y,\phi))\right]\right|\P^{(i+1)}_{o,x_1,\dots,x_i}(\d
\phi)
  \pd^{(i+1)}(\
  x_1,\dots, x_i)\d x_1\dots\d x_i <\infty,
  \label{eq:cond}
\end{equation}
for $i=1,\dots,n+1$. Then 
%for any stationary point process $(\Phi,\P)$,%  \begin{align}
%	\EP \F(\I(y,\Phi)) &= \EP\F(o) + \sum_{i=1}^n
%	\int_{\R^{2i}}\F^{(i)}_{x_1,\dots,x_i}(o) \M^{(i)}(\d
%  x_1,\dots,\d x_i) \nonumber\\
%&+\int_{\R^{2n}}\int_\mathcal{M}
% \F^{(n)}_{x_1,\dots,x_n}(\I(y,\phi))\P^{(n)}_{x_1,\dots,x_n}(\d \phi)
%\M^{(n)}(\d
%  x_1,\dots,\d x_n)
%	\label{eq:asymp}
%  \end{align}
\begin{align}
&  \EP \F(\I(y,\Phi))=
\F(0) + \eta^{-1}\sum_{i=1}^n
	\int_{\R^{2i}}
\E_{\h{x_1},\dots,\h{x_{i}}}\left[\F^{(i)}_{x_1,\dots,x_i}(0)\right]
\pd^{(i+1)}(
  x_1,\dots,x_i) \d x_1,\dots,\d x_i\nonumber\\
  &+\eta^{-1}\int_{\R^{2(n+1)}}\int_M
\E_{\h{x_1},\dots,\h{x_{n+1}}}\left[\F^{(n+1)}_{x_1,\dots,x_{n+1}}(\I(y,\phi))\right]\P^{(n+2)}_{o,x_1,\dots,x_{n+1}}(\d
\phi)
  \pd^{(n+2)}(
  x_1,\dots, x_{n+1}) \d x_1\dots\d x_{n+1}.
	\label{eq:asymp2}
  \end{align}
%  If \eqref{eq:cond} holds for all $i$, then 
%  \[\EP \F(\I(y,\Phi)) =\F(o) + \eta^{-1}\sum_{i=1}^\infty
%\int_{\R^{2i}}
%\E_{\h{x_1},\dots,\h{x_{i}}}\left[\F^{(i)}_{x_1,\dots,x_i}(o)\right]
%\pd^{(i+1)}(
%  x_1,\dots,x_i) \d x_1,\dots,\d x_i.\]
 \end{theorem}
\begin{proof}
The proof follows the  lines of  Theorem 3.1 in \cite{baszczyszyn1997note}. The main difference is that while \cite{baszczyszyn1997note} deals with the FME of $\E \F(\I(y,\Phi))$,  we focus on the reduced Palm version $\EP \F(\I(y,\Phi))$. 
%The result in this theorem can easily be obtained from the proof of Theorem 3.1 in  \cite{baszczyszyn1997note}, by changing the measure from $\P$ to $\PP$. 
We present the proof in the Appendix.
  \end{proof}
 In short, to use FME for interference functionals, it is necessary to verify two things:
  \begin{enumerate}
  \item the functional $\F(\I(y,\Phi))$ is continuous at $\infty$ as in Definition \ref{def:cont}, 
  \item condition \eqref{eq:cond} is valid.
  \end{enumerate}
It is difficult to check condition \eqref{eq:cond}  for a general point
process, and hence we now provide simplified sufficient conditions that are
easy to verify.  
We  begin by providing an easily computable upper bound on
\eqref{eq:cond} whose finiteness can be verified.  The basic idea of the upper bound  on the expansion kernels is simple, and can be easily illustrated for the case $n=1$. The first order expansion kernel is given by
\[\F^{(1)}_{z}(\I(y,\Phi))= \F(\I_z(y,\Phi)+\h{z}\l(z-y))-\F(\I_z(y,\Phi)).\]
This can also be rewritten as
\[\F^{(1)}_{z}(\I(y,\Phi)) = \h{z}\l(z-y)\int_0^1 \F'(\I_z(y,\Phi) +\tau \h{z}\l(z-y)) \d \tau,\]
where $\F'(x) =\frac{\d \F(x)}{\d x}$. So if the derivative of $\F(x)$ is bounded for all $x$, then 
\[|\F^{(1)}_{z}(\I(y,\Phi))|\leq  \|\F'(x)\|_\infty\h{z}\l(z-y).\]
Hence the first order expansion kernel is controlled by the derivative of the function $\F(x)$, and in a similar manner the supremum of the $n$-th order expansion kernel depends on the $n$-th derivative
of the function $\F(x)$, and this is made more precise in the following Theorem.
\begin{theorem}
Let the function $\F(x):\R^+\rightarrow[0,\infty)$ be a smooth function, with derivatives up to
order
$n$ bounded.  For any  $1\leq k \leq n$ and $1\leq p_i \leq n$,  $p_i\in \mathbb{Z}$, $i=1,\dots, k$,  \begin{equation}
	|\F^{(n)}_{x_1,\dots,x_n}(\I(y,\phi))|\leq \FD{k}2^{n-k} \prod_{i=1}^k
	\h{x_{p_i}}\l(x_{p_i} -y),
	\label{}
  \end{equation}
  where 
  $\FD{k} = \left\| \frac{\d^{k}\F(x) }{\d^{k} x} \right\|_\infty$, where $\| \|_\infty$ denotes the $L^\infty$ norm\footnote{$\|f\|_\infty=\inf\left\{a\geq 0: L\left(x:|f(x)|>a\right)=0 \right\}$, and $L$ is the Lebesgue measure  \cite{Folland}.}.
  \label{thm:inq1}
\end{theorem}
\begin{proof}
See Appendix \ref{sec:proofthm2}.
	
\end{proof}
The following corollary combines all the inequalities
provided by
 Threorem \ref{thm:inq1}, and averages the fading.
\begin{cor}
  Let 
  \begin{equation}
  \FE^*_n = \max\{ 2^{n-k}\FD{k},\quad 0\leq k\leq n \}, 
  \label{eq:deriv}
  \end{equation}
  then 
\begin{equation}
|\E_{\h{x_1},\dots\h{x_n}}[\F^{(n)}_{x_1,\dots,x_n}(\I(y,\phi))]|\leq
\FE^{*}_n
  \prod_{i=1}^{n}\int_0^{1} \G\left(\frac{a}{\l(x_i-y)} \right)\d a
 \leq \FE^{*}_n
  \prod_{i=1}^{n}\min\left\{1,\E[\h{}]\l(x_i-y) \right\},
  \end{equation}
  where $\G(x)$ is the CCDF of the fading random variable $\h{}$.
	\label{cor:inq2}
\end{cor}
\begin{proof}
From the definition of $\FE^*_n$, and from Theorem \ref{thm:inq1}, for any
$1\leq
p_i\leq n$, and $k \leq n$,
  \begin{equation*}
	|\F^{(n)}_{x_1,\dots,x_n}(\I(y,\phi))|\leq \FE^*_n \prod_{i=1}^k
	\h{x_{p_i}}\l(x_{p_i} -y).
	\label{}
  \end{equation*}
Since $\min\{1,b\}\min\{1,a\} = \min\{1,a,b,ab\}$, and using the above inequality for all $k$ and all product combinations of $x_i$, we obtain\begin{equation*}
|\F^{(n)}_{x_1,\dots,x_n}(\I(y,\phi))|\leq \FE^*_n \prod_{i=1}^n
\min\left\{1,
	\h{x_{i}}\l(x_{i} -y)\right\}.
	\label{}
  \end{equation*}
Since
$|\E_{\h{x_1},\dots,\h{x_n}}[\F^{(n)}_{x_1,\dots,x_n}(\I(y,\phi))]|\leq
\E_{\h{x_1},\dots,\h{x_n}}[|\F^{(n)}_{x_1,\dots,x_n}(\I(y,\phi))|]$, and
from the
independence of $\h{x_i}$,
\begin{equation}
|\E_{\h{x_1},\dot,\h{x_n}}[\F^{(n)}_{x_1,\dots,x_n}(\I(y,\phi))]|\leq \FE^*_n \prod_{i=1}^n\E_{\h{x_i}}
\min\left\{1,
	\h{x_{i}}\l(x_{i} -y)\right\}.
	\label{eq:999}
  \end{equation}
The average 
\begin{align*}
  \E_{\h{x_i}} [\min\left\{1,
	\h{x_{i}}\l(x_{i} -y)\right\}]=& \int_0^\infty \P\left( \min\left\{1,
	\h{x_{i}}\l(x_{i} -y)\right\}>a \right)\d a\\
	=&\int_0^1 \G\left(	\frac{a}{\l(x_{i} -y)}\right)\d a.
\end{align*}
Alternatively, using the fact $\E[\min\{x,y\}]\leq \min\{\E [x],\E[y]\}$ and
from
\eqref{eq:999} we obtain the
other bound. 
\end{proof}
The continuity of the functional at infinity is simpler and is stated in the next lemma.
\begin{lemma}
If $\F(x): \R^+\rightarrow \R^+$ is a continuous function, then the functional $\F(\I(y,\Phi))$ is continuous at infinity as in Definition \ref{def:cont}.
\end{lemma}

\begin{proof}
From the continuity of $\F$ we have
\[\lim_{\|z\|\rightarrow \infty} \F\left(\sum_{\x \in \Phi\cap B(o,\|z\|)}\h{\x}\l(\x-y)\right) = \F\left( \lim_{\|z\|\rightarrow \infty}\sum_{\x \in \Phi\cap B(o,\|z\|)}\h{\x}\l(\x-y)\right),\]
and the result follows from the monotone convergence theorem \cite{Folland}. 
\end{proof}

\section{FME of Success Probability}
\label{sec:out}
In this section, we obtain a series expansion of the success probability
using
Theorem \ref{thm:main}. Each node $x\in \Phi$ is associated with a
receiver
$r(x)$ at a distance $R$ in a random direction.
The success probability of a typical source destination link is
\begin{equation}
 \p = \PP(\sinr(o,r(o)) >\T),
 \label{eq:outage}
\end{equation}
where,
\[ \sinr(o,r(o)) = \frac{\ps\l(R)}{\I(r(o),\Phi)+\W},\]
and $\W$ is the noise power at the receiver, and $\ps$ is the fading\footnote{$\ps$  may also represent the power received by  from a source by its destination at a distance $x$  such that $\l(x)=1$. } between the source destination pair.
For this paper we assume $\ps$ is independent of $\I(r(o))$, although
the
case of $g(\I(r(o))$, where $g(x)$ is a random function dependent on
$\ps$,
can be dealt with in a similar manner. Hence the success probability is
\begin{align}
   \p &= \PP(\sinr(o,r(o)) >\T)\nonumber\\
    &=\PP\left(\frac{\ps\l(R)}{\I(r(o),\Phi)+\W} >\T\right)\nonumber\\
	&= \EP\F(\v\W + \v\I(r(o),\Phi)),
  \label{eq:main}
\end{align}
where $\v =\T/\l(R)$ and $\F(x)$ is the CCDF of the random variable
$\ps$.
The following  theorem  provides a series expansion of $\p$.
\begin{theorem}
\label{thm:main1}
Let $\Phi$ be a stationary point process of transmitters of density $\eta$  such that the product densities are finite and such that
%  If for any stationary distribution of transmitters of density $\eta$,
  \begin{equation}
\eta^{-1}\FE^{*}_i\int_{\R^{2i}}\prod_{k=1}^{i}\min\left\{1,\E[\h{}]\l(x_k-r(o))
\right\} \pd^{(i+1)}(
  x_1,\dots, x_i)\d x_1\dots\d x_i <\infty,
  \label{eq:cond23}
\end{equation}
for $1\leq i\leq n+1$, where 
$  \FE^*_i = \max\left\{ 2^{i-k}\left\|\frac{\d^k \F(\v\W+\v x)}{\d ^kx}\right\|_\infty,\quad 0\leq k\leq i \right\}$.
  Then,
\begin{align}
\p =&\F(\v\W) + \eta^{-1}\sum_{i=1}^n
	\int_{\R^{2i}}
\E_{\h{x_1},\dots,\h{x_{i}}}\left[\F^{(i)}_{x_1,\dots,x_i}(\v\W)\right]
\pd^{(i+1)}(
	x_1,\dots,x_i) \d x_1,\dots,\d x_i+\err{n},
	\label{eq:success}
\end{align}
and,
%\begin{align}
%  \err{n} =&\eta^{-1}\int_{\R^{2(n+1)}}\int_\mathcal{M}
%\E_{\h{x_1},\dots,\h{x_{n+1}}}\left[\F^{(n+1)}_{x_1,\dots,x_{n+1}}(\v\W+\v\I(r(o),\phi))\right]\P^{(n+2)}_{o,x_1,\dots,x_{n+1}}(\d
%\phi)
%  \pd^{(n+2)}(
%  x_1,\dots, x_{n+1}) \d x_1\dots\d x_{n+1}.
%  \end{align}
%  Also
\[|\err{n}| \leq
\frac{\eta^{-1}\FE^{*}_{n+1}}{(n+1)!}\int_{\R^{2(n+1)}}\prod_{k=1}^{n+1}\min\left\{1,\E[\h{}]\l(x_k-r(o))
\right\}
  \pd^{(n+2)}(
  x_1,\dots, x_{n+1})\d x_1\dots\d x_{n+1}.\]
\end{theorem}
\begin{proof}
  Follows from Theorem \ref{thm:main} and Corollary \ref{cor:inq2}. The $(n+1)!$ in the error term results from integrating over the entire domain and not a cone.
\end{proof}
The condition \eqref{eq:cond23} is  a sufficient condition that
is easy to
verify, but not a necessary one. 
In the following subsections we provide various examples to illustrate the application of Theorem \ref{thm:main1}. 
\subsection{Poisson point process (PPP)} 
As noted in Section \ref{sec:sys_model},  the interference in a network with PPP distribution of nodes has been analyzed extensively, and the success probability has been obtained for different fading models. In this subsection we shall compare the existing results with the approximation obtained by the FME. We assume that the underlying transmitter nodes $\Phi$ form a PPP of density $\eta$. For a PPP, the product densities are \cite{stoyan} 
\[\pd^{(n+1)}(x_1,\hdots, x_n)= \eta^{n+1},\]
and can easily be obtained from the independence properties. As mentioned earlier, Rayleigh fading, \ie, $\h{} \sim \exp(1)$, is easy to deal with analytically using the Laplace transform of the interference. We begin with the FME of $\p$ in a PPP network with Rayleigh fading.
\subsubsection{Rayleigh Fading}
When the small-scale fading is Rayleigh distributed, $\h{}$ and $\ps$ are exponentially distributed.
Since $\ps$ is exponentially distributed with unit mean, from \eqref{eq:main} the outage probability is given by
\[\p = \EP \exp\left(-\v\W-\v \I(r(o),\Phi) \right),\] 
and hence $\F(x)=\exp(-\v\W-\v x)$.
Hence to evaluate the FME in \eqref{eq:success}, it is necessary to evaluate the average of the $n$-th order expansion kernels with respect to fading $\E_{\h{z_1},\dots,\h{z_{n}}}[\F^{(n)}_{z_1,\hdots,z_n}(0)]$.  For $n=1$,
\[\E_{\h{z_1}}[\F^{(n)}_{z_1}(0)] = \E_{\h{z_1}}  \exp\left(-\v\W- \v\h{z_1}\l(x-r(o))\right) -  \exp\left(-\v\W \right),\]
and since $\h{z_1}$ is exponentially distributed, 
\[ \E_{\h{z_1}}[\F^{(n)}_{z_1}(0)] = \frac{  \exp(-\v\W)}{1+\v\l(z_1-r(o))}- \exp\left(-\v\W \right)=\frac{ - \exp(-\v\W)}{1+\v^{-1}\l(z_1-r(o))^{-1}}.\]
Similarly, the $n$-th order expansion kernel can be easily shown to be equal to
\begin{equation}
\	\E_{\h{z_1},\dots,\h{z_{n}}}[\F^{(n)}_{z_1,\hdots,z_n}(0)]=
\left\{\begin{array}{ll}
  (-1)^{n}\exp(-\v\W)\prod_{i=1}^{n}\Delta(z_i),&
 z_n\preccurlyeq,\hdots,\preccurlyeq z_1,\\
0,&\text{otherwise},
\end{array}\right.
\end{equation}
where 
\[\Delta(x)=\frac{1}{1+\v^{-1}\l(x-r(o) )^{-1}}.\]
To use Theorem \ref{thm:main1}, we have to first verify the finiteness of the integral in \eqref{eq:cond23}. Since $\F(x)=\exp(-\v x-\v\W)$, all its derivatives are well behaved and bounded and hence $\FE^{*}_i <\exp(-\v\W)(2\v)^i$ for all $i$. The condition \eqref{eq:cond23} can be simplified to checking  the finiteness of  
\[\eta^{i}\F^{*}_i \left(\int_{\R^{2}}\min\left\{1,\l(x)
\right\}\d x\right)^i,\]
which is finite since $\min\{1,\l(x)\}$ is a well behaved function. Hence using Theorem \ref{thm:main1}, it follows that
\begin{align}
\p =& \exp(-\v\W)\sum_{i=0}^n
\frac{(-\eta)^{i}}{i!}\left(\int_{\R^{2}}
	 \Delta(x)  \d x\right)^i+\err{n}\nonumber.
\end{align}
We get the $i!$ in the denominator since $
\E_{\h{z_1},\dots,\h{z_{n}}}[\F^{(n)}_{z_1,\hdots,z_n}(0)]$ is defined
only on the simplex $ z_n\preccurlyeq,\hdots,\preccurlyeq z_1$, while the
integration is  over the complete domain. The error term is bounded by
\begin{equation}|\err{n}|\leq \frac{ (2\v\eta)^{n+1} }{(n+1)!}\exp(-\v\W)\left(\int_{\R^{2}}\min\left\{1,\l(x)
\right\}\d x\right)^{n+1} =\Theta(\eta^{n+1}).\end{equation}
The error term tends to zero as $n\rightarrow 0$ because of the $(n+1)!$ in the denominator.  We also observe that the error term gets smaller as the noise $\W$ increases.  This is because as the noise power increases, the interference has smaller influence on the outage probability, and a fewer terms of the FME suffice for a good approximation.
\begin{figure}
\centering
\includegraphics[width=10cm]{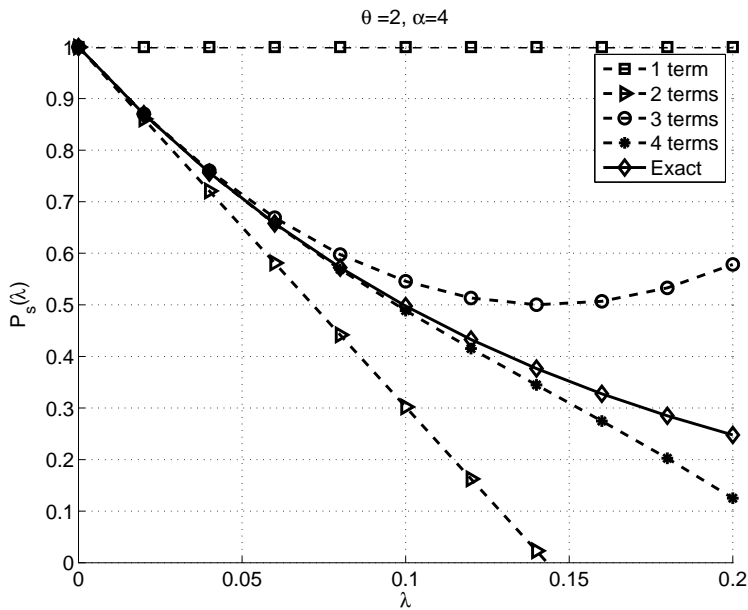}
\caption{Illustration of the FME approximation with increasing number of terms, for $\alpha=4$ and $\v=1$. We observe that the approximation gets better with increasing terms.}
\label{fig:rayleigh_approx}
\end{figure}
The exact outage analysis in a Poisson network is obtained in \cite{bacelli-aloha} using the probability generating functional, and the success probability is 
\[\p = \exp(-\v\W) \exp\left(-\eta \int_{\mathbb{R}^2} \Delta(x)\d x\right)=\exp(-\v\W)\sum_{i=0}^\infty
\frac{(-\eta)^{i}}{i!}\left(\int_{\R^{2}}
	 \Delta(x)  \d x\right)^i,\]
and we observe that FME provides the exact expansion.  Analytically for a PPP, the difficult part of the FME is to obtain the expansion kernels since the number of terms in the  $n$-th order expansion kernel  grows exponentially with $n$.  In Fig. \ref{fig:rayleigh_approx}, the FME is plotted for different orders, \ie, $n$ for $\alpha=4$, $\W=0$, and $\nu=1$.  
 
\subsubsection{Nakagami-m fading}
In the case of Nakagami-$m$ fading, the CCDF of the  fading is
\[\F(x)= \frac{\Gamma(m,mx)}{\Gamma(m)}.\]
The expansion kernels do not have a simple form as in the case of Rayleigh fading. We now obtain
the expansion kernels up to order $2$ for the case of the singular path loss model $\l(x)=\|x\|^{-\alpha}$, and no noise, \ie, $\W=0$. For the other cases, numerical methods have to be employed to obtain the expansion kernels.
\begin{equation}\E_{\h{z_1}}[\F^{(1)}_{z_1}(0)] = \E_{\h{z_1}}  \F\left(\v\h{z_1}\l(x-r(o))\right) -  1.\end{equation}
So the first order term of the FME series is given by 
\[\eta\int_{\R^2}  \left[\E_{\h{z_1}}  \F\left(\v\h{z_1}\|(x-r(o)\|^{-\alpha}\right) -  1 \right]\d x.\]
Using simple substitutions it is equal to
\[\E[\h{}^{2/\alpha}]\eta\int_{\R^2}  \left[\F\left(\v\|x\|^{-\alpha}\right) -  1 \right]\d x=\frac{\pi\Gamma\left( m-\frac{2}{\alpha}\right)\Gamma\left( m+\frac{2}{\alpha}\right)}{\Gamma(m)^2}.\]
So the FME with two terms  is 
\begin{equation}
\p = 1-\eta\v^{-\frac{2}{\alpha}}\frac{\pi\Gamma\left( m-\frac{2}{\alpha}\right)\Gamma\left( m+\frac{2}{\alpha}\right)}{\Gamma(m)^2} +\err{1}.
\end{equation}
In this case, the error is bounded by
\[|\err{1}|\leq  2(\v\eta)^{2} \left(\int_{\R^{2}}\min\left\{1,\l(x)
\right\}\d x\right)^{2}.\]
This error bound can be improved by using the alternative bound in Corollary \ref{cor:inq2}.
More terms in the series can be computed to have a better approximation of the error probability.
For example the term corresponding to $\E_{\h{z_1}\h{z_2}}[\F^{(2)}_{z_1,z_2}(0)]$ for different $m$ are given in Table \ref{tab:tab2} and the term corresponding to  $\E_{\h{z_1}\h{z_2}\h{z_3}}[\F^{(3)}_{z_1,z_2,z_3}(0)]$ are given in Table \ref{tab:tab3}. In Fig. \ref{fig:naka_approx}, the approximations provided by FME with different number of terms is plotted for $m=2$ and $m=3$.
\begin{table}
\begin{center}
\begin{tabular}{|l|l|}
\hline
$m=2$ &$\frac{2 \pi ^4 (\alpha -4) (\alpha +2)^2 \csc ^2\left(\frac{2 \pi }{\alpha }\right)}{\alpha ^3}\lambda^2\v^{-\frac{4}{\alpha}}$\\ \hline
$m=3$&$\frac{2\pi ^4 (\alpha -4) (\alpha -2) (\alpha +1)^2 (\alpha +2)^2 \csc ^2\left(\frac{2 \pi }{\alpha }\right)}{\alpha
   ^6}\lambda^2\v^{-\frac{4}{\alpha}}$\\\hline
$m=4$&$\frac{2 \pi ^4 (\alpha -4) (\alpha -2) (\alpha +1)^2 (\alpha +2)^2 (3 \alpha -4) (3 \alpha +2)^2 \csc ^2\left(\frac{2
   \pi }{\alpha }\right)}{27 \alpha ^9}\lambda^2\v^{-\frac{4}{\alpha}}$\\\hline
\end{tabular}
\end{center}
\caption{ The third term in the FME, \ie, $\frac{\lambda^2}{2}\int_{\R^2}\int_{\R^2}\E_{\h{z_1}\h{z_2}}[\F^{(2)}_{z_1,z_2}(0)]\d z_1 \d z_2$ for different $m$}
\label{tab:tab2}
\end{table}

\begin{table}
\begin{center}
\begin{tabular}{|l|l|}
\hline
$m=2$ &$ \frac{-4 \pi ^6 (\alpha -6) (\alpha +2)^3 \csc ^3\left(\frac{2 \pi }{\alpha }\right)}{3\alpha ^4} \lambda^3\v^{-\frac{6}{\alpha}}$\\ \hline
$m=3$&$\frac{-4 \pi ^6 (\alpha -6) (\alpha -3) (\alpha +1)^3 (\alpha +2)^3 \csc ^3\left(\frac{2 \pi }{\alpha }\right)}{3\alpha
   ^8}\lambda^3\v^{-\frac{6}{\alpha}}$\\\hline
$m=4$&$\frac{4 \pi ^6 (\alpha -6) (\alpha -3) (\alpha -2) (\alpha +1)^3 (\alpha +2)^3 (3 \alpha +2)^3 \csc ^3\left(\frac{2
   \pi }{\alpha }\right)}{81 \alpha ^{12}}\lambda^3\v^{-\frac{6}{\alpha}}$\\\hline
\end{tabular}
\end{center}
\caption{ The fourth term in the FME, \ie, $\frac{\lambda^3}{6}\int_{\R^2}\int_{\R^2}\int_{\R^2}\E_{\h{z_1}\h{z_2}\h{z_3}}[\F^{(3)}_{z_1,z_2,z_3}(0)]\d z_1 \d z_2\d z_3$ for different $m$}
\label{tab:tab3}
\end{table}
\begin{figure}
\centering
\includegraphics[width=0.5\columnwidth]{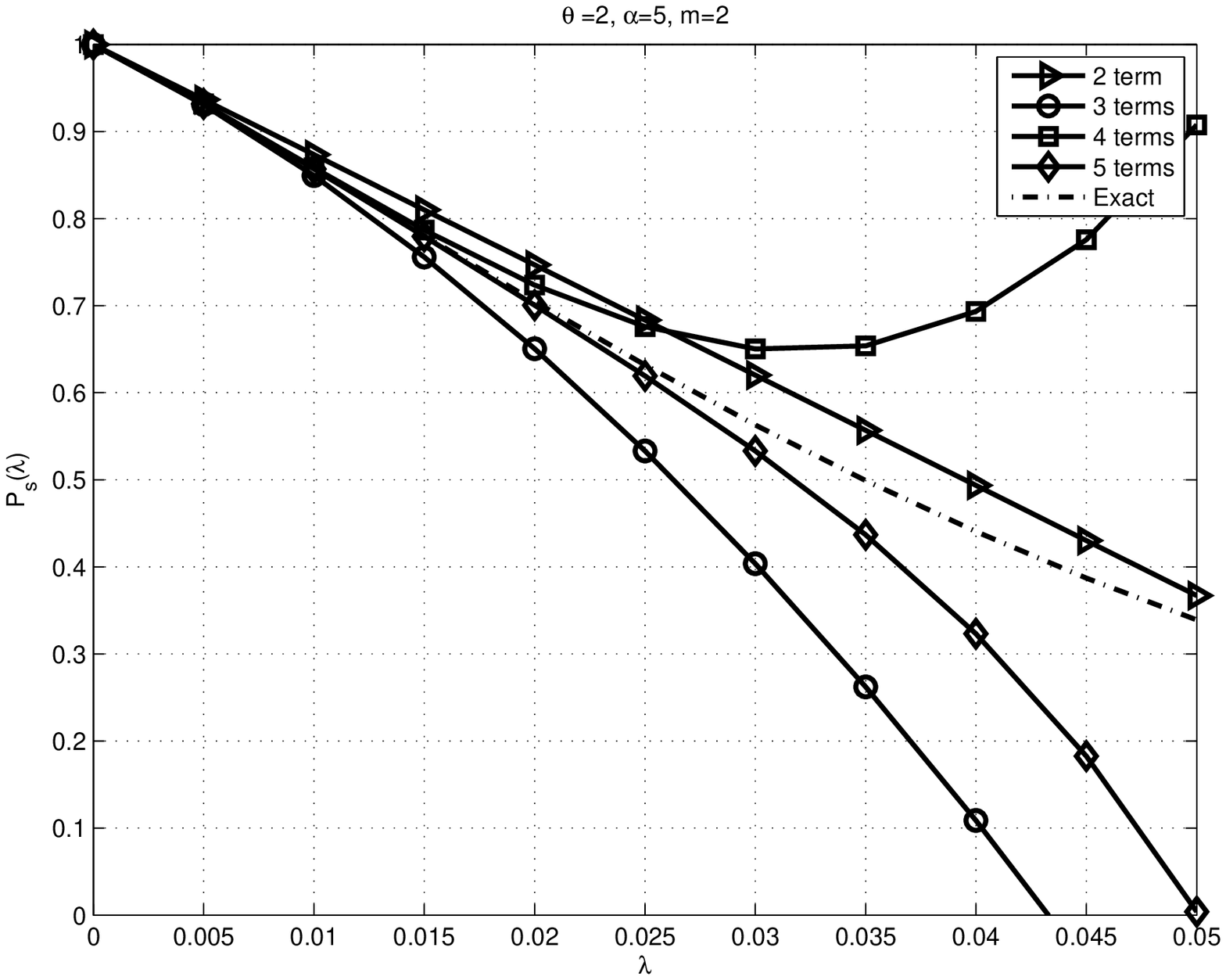}\includegraphics[width=0.5\columnwidth]{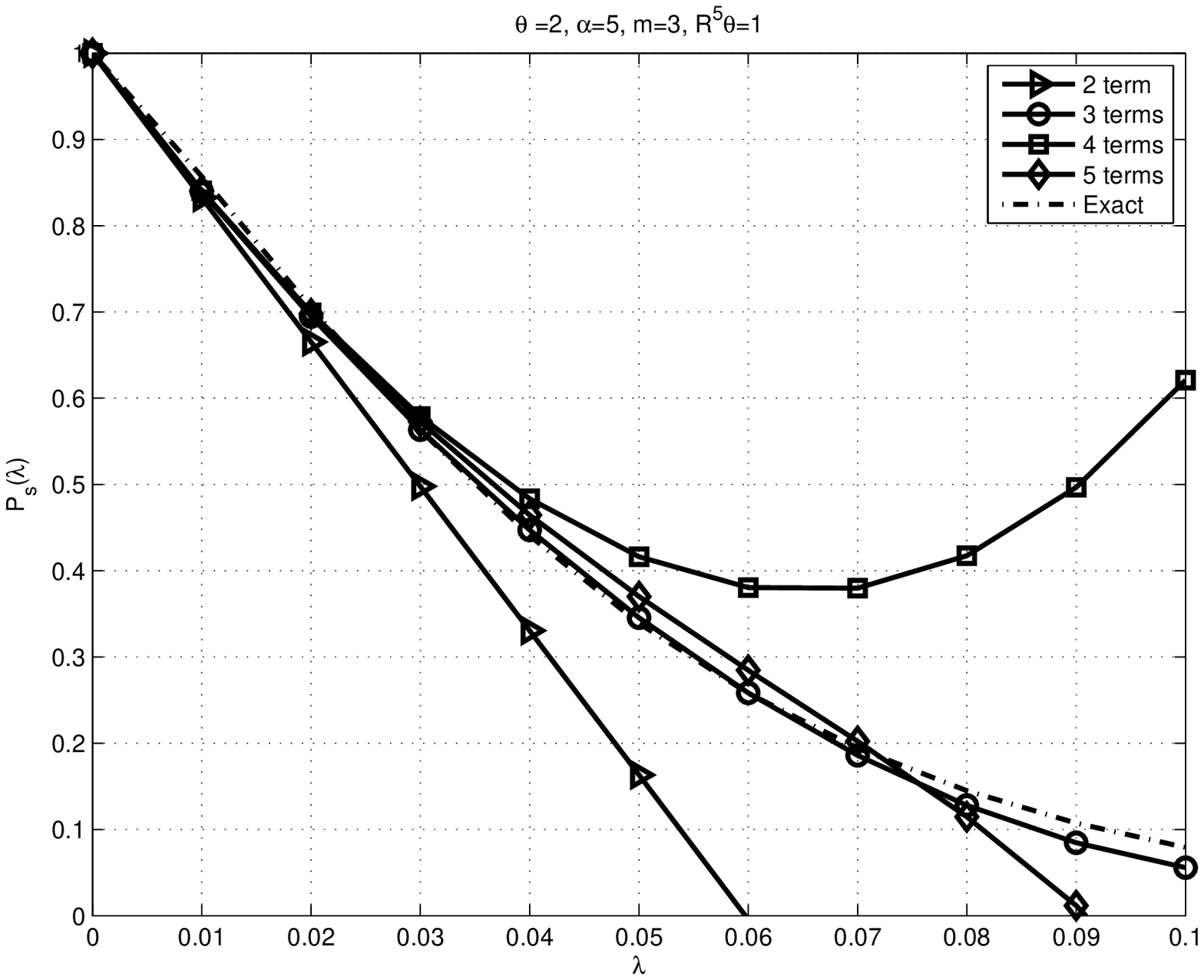}
\caption{Illustration of the FME approximation with increasing number of terms, for $\alpha=5$ and $\v=1$. The left figure corresponds to $m=2$ while the right figure corresponds to $m=3$.}
\label{fig:naka_approx}
\end{figure}
\subsubsection{Log-Normal Shadowing}
Log-Normal shadowing is  commonly used to model large scale fluctuations in the channel, and  the CCDF of the fading is given by
\[\F(x)=\frac{1}{2}-\frac{1}{2} \text{erf}\left(\frac{\log (x)-\mu}{\sqrt{2\sigma^2}}\right).\]
Neglecting noise, we obtain the first term of the FME as 
\[\p=1-2\pi e^{\frac{4 \sigma^2 -4 \alpha  \mu }{\alpha ^2}}\v^{2/\alpha}\eta+\err{2}.\]
The other terms have to be obtained using numerical methods. 

\subsection{Outage in the Mat\'ern hard core process (CSMA)}
The spatial  distribution  of the transmitters that concurrently
 transmit in  a CSMA network is difficult to  determine, but the transmitting set
can be closely approximated by a modified Mat\'ern hard-core processes
\cite{baccelli2009stochastic}. We start with a  Poisson point process  $\Psi$ of
unit density. To each node $\x \in \Psi$, we associate a mark $m_\x$, a uniform random
variable in $[0,1]$. The contention neighborhood of a node $x$ is the set of
nodes which result in an interference power of  at least    $\pp$ at $\x$, \ie,
\begin{equation}\bar{\calN}(\x) = \{\y \in \Psi: \h{\y \x} \l(\y -\x) > \pp \}.\end{equation}
A node $\x \in \Psi$ belongs to the final CSMA transmitting set if
\[ m_\x < m_\y,\quad \forall \y\in \calN(\x). \]
The average number of nodes in the contention neighborhood of $\x \in \Psi$,
does not depend on the location $\x$ by the stationarity of $\Psi$ and is equal
to \cite{baccelli2009stochastic}
\[\calN = \E[\bar{\mathcal{N}}(x)]= 2\pi \int_0^\infty \F\left(\pp \l(x)^{-1}\right)  \d x,\]
where $\F(x)$ is the CCDF of the fading distribution.
The density of the modified Mat\'ern process $\Phi$ is equal to
\begin{equation}
\lambda =   \left(\frac{1-\exp(-{\calN})}{{\calN}}\right).\end{equation}
The next Lemma from \cite{baccelli2009stochastic}, provides the higher order product densities of the 
CSMA transmitting set.
\begin{lemma}
The $n$-th order product density of the modified Mat\'ern hard core process is
\[\pd^{(n)}(y_1,....,y_{n-1})=  n!\int_{[0,1]^n}\i(0\leq t_1\leq\hdots\leq t_n\leq 1) f(t_1,\hdots,t_n)\d t_1\hdots \d t_n,\]
where 
\begin{align*}
f(t_1,\hdots,t_n)=& \exp\left\{-\sum_{J\subset\{1,..,n\}}(-1)^{\#J+1}t_{\min_{i\in J}}\int_{\R^2} \prod_{i\in J}\F(\pp\l(x-y_i)^{-1})\d x \right\}\\ &\cdot\prod_{j=1}^n\prod_{i<j}\left(1-\F\left(\pp\l(y_i-y_j)^{-1}\right)\right),
\end{align*}
with the convention $y_n=0$.
\end{lemma}
%\begin{proof}
%See \cite{baccelli2009stochastic}.
%\end{proof}
Using these product densities, a simple sufficient condition for the FME to hold true is provided in the next Lemma.
\begin{lemma}
\label{lem:mat}
If $\FE_i^* <\infty$, then the FME expansion holds true for the modified
Mat\'ern hard core process and 
  \begin{equation}
|\err{n}| \leq \frac{\mathcal{N} \FE^{*}_{n+1}  }{(n+1)!(1-\exp(-\mathcal{N}))}\left(\int_{\R^{2}}\min\left\{1,\E[\h{}]\l(x-r(o))
\right\} (1-\F(\pp\l(x))) \d x\right)^{n+1}. 
\end{equation}
\end{lemma}
\begin{proof}
It is easy to observe that 
\[\pd^{(n)}(y_1,....,y_{n-1})\leq   \prod_{k=1}^{n-1} 1- \F(\pp\l(y_k)),\]
 and hence the result follows from \eqref{eq:cond23}.
\end{proof}
Observe that as $\pp \rightarrow 0$, the density of the transmitter set decreases, since the contention set increases. In Lemma \ref{lem:mat}, we observe that the rate of decay of the error with respect to $\pp$ depends on the behavior of $\F(x)$ at the origin.
Using Theorem \ref{thm:main1}, and Lemma \ref{lem:mat}, the outage probability
can be obtained for any fading distribution when the TXs locations are modeled
by the CSMA-Mat\'ern process.

\subsection{Determinantal Point Processes (DPP)}
DPPP were introduced  by O. Macchi and the points of a DPP exhibit soft repulsion.  Hence  these processes are particularly appealing for modeling node locations in structured spatial networks \eg, cellular networks. DPP  have been used to analyze eigenvalues of random matrices, zeros of analytical functions, and Fermionic gases (in physics) and exhibit a rich mathematical structure. See \cite{soshnikov2000determinantal,hough2006determinantal} for a good exposition of DPP.  DPP is particularly suited  to FME analysis since it is defined by its product densities.
A point process is determinantal if its $n$-th order product density is given by
\[\rho^{n}(x_1,\hdots,x_n) =  \det(\K(x_i,x_j))_{1\leq i,j\leq n}.\]
$\K(x,y)$ is the kernel of the DPP and is assumed to be locally square integrable, a local trace class operator, hermitian and non-negative definite.
%
%properties (these might be an
%overkill):
%\begin{description}
% \item[A.1] Locally square integrable:
% \begin{equation}
%  \int_{C^2} |\K(x,x)|^2 \d x\d x < \infty,
% \end{equation}
% where $D$ is a compact subset of $\R^2$.
% \item[A.2] Locally trace class operator: $\K(x,y)$ determines a linear
%bounded operator from $L^2(\R^2)$ to $L^2(\R^2)$ by
%\[\mathcal{K}f:= \int \K(x,y)f(y) \d y.\]
%Indeed by Assumption A.1, it follows that $\mathcal{K}$ is a compact
%operator
%on $L^2(\R^2)$. Also by restricting the operator $\mathcal{K}$ onto a
%compact
%subset $D$ we obtain a new operator $\mathcal{K}_D$ which is also a
%compact
%operator. Hence its spectrum is discrete and has finite multiplicity. Werequire the operator $\mathcal{K}_D$ to be locally trace class, \ie,
%\[\sum_i \lambda^D_i <\infty,\]
%for all compact $D\subset \R^2$.
%\item[A.3] $\K$ is hermitian, \ie,
%\[\K(x,y) = \overline{\K(y,x)}.\]
%\item [A.4] Non negative definite: For any integer $n$, the matrix
%\[\{\K(x_i,x_j)\}_{1\leq i,j\leq n}\]
%is non-negative definite.
%\[\int \K(x,y) f(x)f(y) > 0\, \forall f \in L^2(\R^2).\]
%\end{description}
%\subsection{Stationary DPP }
Stationarity of a point process implies that the product measures are
translation invariant, \ie,
\begin{equation}\rho^{(n)}(x_1+y,\hdots,x_n+y)=\rho^{(n)}(x_1,\hdots,x_n).\end{equation}
Hence for a stationary DPP  the kernel should
be of the form
\[\K(x,y)= \K(x-y,0) := \K(x-y).\]
%So for a stationary DPP, the corresponding integral operator defined by the kernel is the
%convolution
%operator, \ie,
%\[\mathcal{K}f= (\K*f)(x),\]
%and hence it tie in well with the Fourier theory.
%A stationary DPP is also ergodic and mixing (of any order)
%\cite{soshnikov2000determinantal}.
The average number of points in a set $B\subset \R^2$ is equal to
\[\E[\Phi(B)] = \int_B \rho^{(1)}(x) \d x,\]
and for a stationary DPP, we can observe that the density is equal to
\begin{equation} \eta = \frac{\E[\Phi(B)]}{|B|} = \K(0). \end{equation}
We have the following upper bound on the product densities of the DPP.
%\begin{lemma}
%The $n$-th order product density of a stationary DPP is upper bounded by \[\rho^{(n)}(x_1,\hdots,x_n) \leq \K(0)^n. \]
%\end{lemma}
%\begin{proof}
% Follows from Hadamard's inequality \cite{thomas}.
%\end{proof}
\begin{lemma}
For a stationary DPP, if $\K(x)=\K(-x)$,
 \[\rho^{(n)}(x_1,\hdots,x_n)= \rho^{(n)}(y_1,\hdots,y_{n-1},0) \leq
\frac{\prod_{k=1}^{n-1}\rho^{(2)}(y_k)}{\K(0)^{n-2}}, \]
where $y_k=x_k-x_n$.
\label{lem:inq2}
\end{lemma}
\begin{proof}
 Follows from Fan's inequality  \cite{fan1955some}.
\end{proof}
It also follows from the non-negative definitive nature of the kernel
(A.4) that $\K(x) <\K(0)$.
We now provide a few examples of stationary DPP.
\begin{enumerate}
\item{\em The Ginibre Ensemble} It was proved in
\cite{ginibre1965statistical}
by Ginibre that the eigenvalues
of complex non-hermitian $n\times n$ matrices with unit Gaussian random
variables, form a DPP with kernel (in the limit $n \rightarrow \infty$)
\[\K(z_1,z_2)=\pi^{-1}\exp\left( -\|z_1-z_2\|^{2}/2\right). \]
The Ginibre ensemble is a stationary DPP (observe that $\K(z_1, z_2) =
\K(z_1-z_2,0)$) of density $\pi^{-1}$. The density of the point process can be modified (decreased) by changing the variance of the Gaussian random variables and is equal to $(\pi\sigma^2)^{-1}$. See Figure \ref{fig:compare} for a comparison between a PPP and a DPP.
\begin{figure}
\centering
\includegraphics[width=0.5\columnwidth]{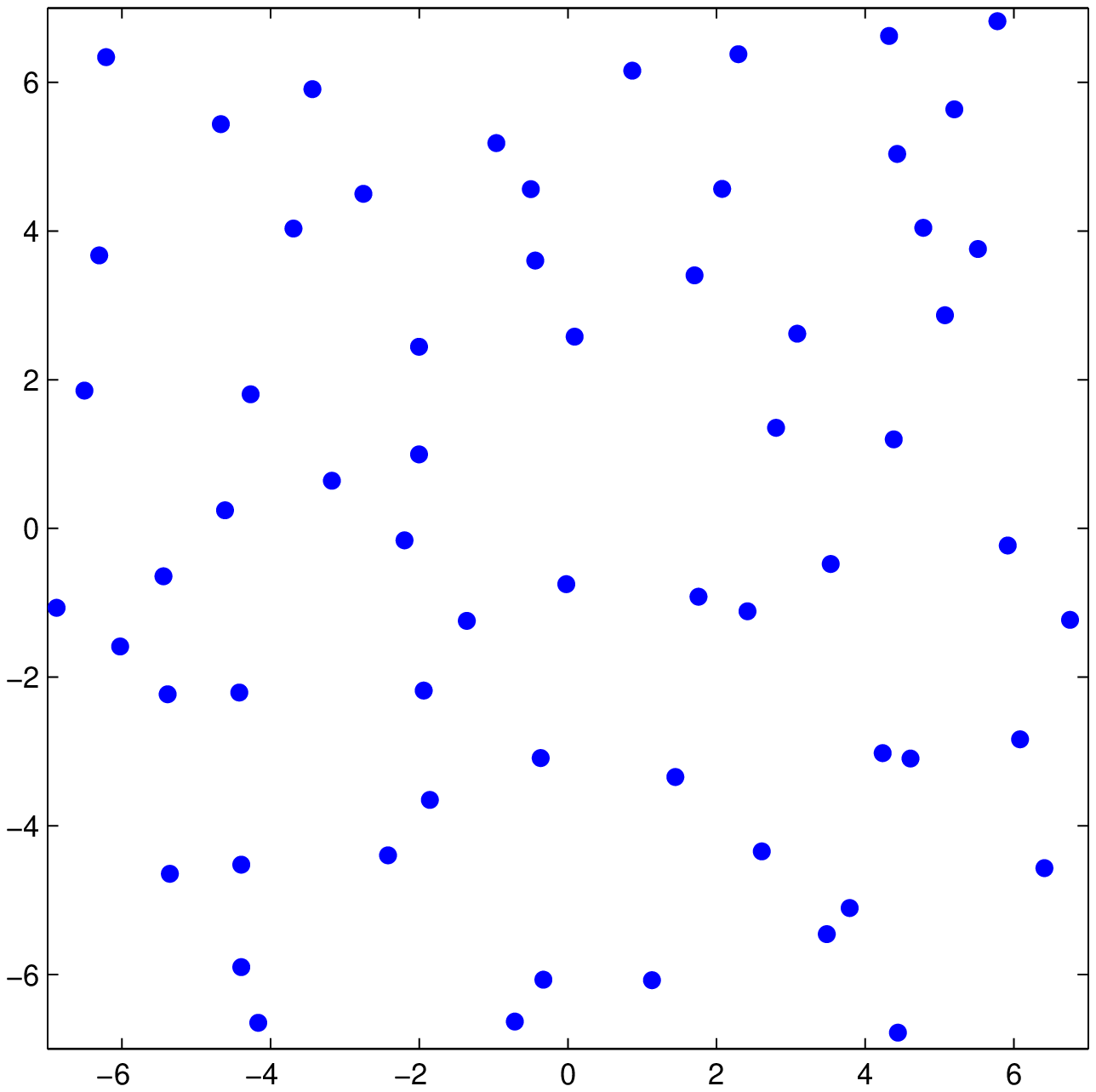}\includegraphics[width=0.5\columnwidth]{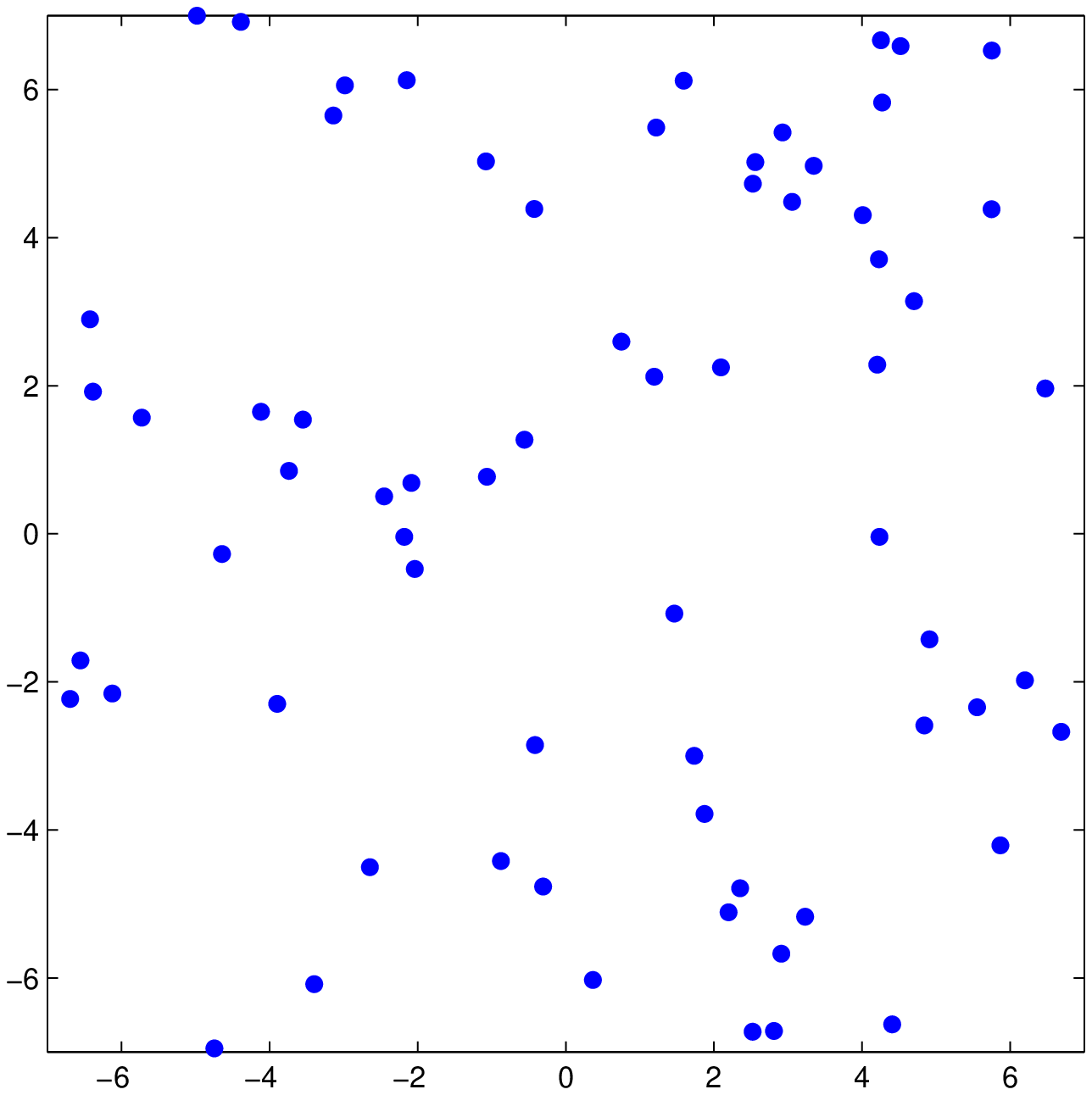}
\caption{The left figure illustrates a realization of the Ginibre DPP, while the right figure corresponds to a PPP with the same density. Observe that the points in the DPP seem more regular compared to the PPP.}
\label{fig:compare}
\end{figure}
\item  The eigenvalues of $A=\tilde{A}+iv\tilde{B}$, where $\tilde{A}$,
$\tilde{B}$
belong to the Gaussian unitary ensemble,  and $v<1$, form a DPP with kernel,
   \[\K(z_1,z_2) =\frac{1}{\pi (1-\tau^2)}\exp\left(-\frac{
   \|z_1-z_2\|^2}{2(1-\tau^2)}\right),\]
   where $\tau =(1-v^2)/(1+v^2) \in [0,1]$.
Observe that the modified Ginibre DPP is of density
$1/(\pi(1-\tau^{2}))$.
\item { Fermionic Gas (Sine DPP)} The probability distribution of
fermion
($n\rightarrow \infty$) locations on the real line is a DPP
\cite{soshnikov2000determinantal}
with kernel \[\K(x_1,x_2) = \frac{\sin(\pi(x_1-x_2))}{\pi(x_1-x_2)}.\]
This has been extended to higher dimensions in \cite{torquato2008point}, and for two-dimensions, the kernel is given by
\[\K(r) = \frac{J_1(2\sqrt{\pi}r)}{2\sqrt{\pi}r},\]
where $r=\|x_1-x_2\|$ and $J_1(x)=\pi^{-1}\int_0^\pi \cos(x \sin(\theta) -\theta)\d \theta$ is the Bessel function of the first order.
The density of the Sine DPP is $\K(0)=1/2$.
\end{enumerate}

The following theorem simplifies the necessary condition of the FME for the case of a DPP.
\begin{lemma}
  \label{lem:dpp_out}
  If the transmitters form a stationary and isotropic DPP with kernel $\K(x)$ and
  $\FE^*_{n+1}<\infty$, 
then the FME holds up to $n$ terms and the error is bounded by
 \begin{equation}
|\err{n}| \leq \frac{\FE^{*}_{n+1}\eta^{n+2}}{(n+1)!}\left(\int_{\R^{2}}\min\left\{1,\E[\h{}]\l(x-r(o)) \right\}
\left(1-\frac{\K^2(x)}{\K^2(0)} \right)\d x\right)^{n+1} <\infty. 
\end{equation}
\end{lemma}
\begin{proof}
From Theorem \ref{thm:main1}, it suffices to show that \eqref{eq:cond23}
holds
  true. From Lemma \ref{lem:inq2}, we also have
\[ \rho^{(n)}(y_1,\hdots,y_{n-1}) \leq
\frac{\prod_{k=1}^{n-1}\rho^{(2)}(y_k)}{\K(0)^{n-2}},\]
and hence \eqref{eq:cond23} is 
\begin{align}
\eta^{-1}\FE^{*}_i\int_{\R^{2i}}\prod_{k=1}^{i}\min\left\{1,\E[\h{}]\l(x_k-r(o))
\right\} \pd^{(i+1)}(
  x_1,\dots, x_i)\d x_1\dots\d x_i,\nonumber\\
  <
\eta^{-1}\FE^{*}_i\int_{\R^{2i}}\prod_{k=1}^{i}\min\left\{1,\E[\h{}]\l(x_k-r(o))
\right\} \frac{\prod_{k=1}^{i}\rho^{(2)}(x_k)}{\K(0)^{i-2}} \d x_1 \dots
\d
  x_i.\nonumber\\
  =\eta^{-1}\FE^*_i \K(0)^2\left(\frac{1}{\K(0)}\int_{\R^{2}}
  \min\left\{1,\E[\h{}]\l(x-r(o)) \right\}\rho^{(2)}(x) \d
  x\right)^{i}.
  \label{eq:1234}
\end{align}
Since for a stationary DPP,
\[\pd^{(2)}(x)=\K^2(0)-\K^2(x),\]
and $\eta=\K(0) > \K(x)$. 
Hence  $\pd^{(2)}(x) \leq \K^2(0)$ and hence \eqref{eq:1234} is always finite.  We obtain the error bound by substituting $\pd^{(2)}(x)=\K^2(0)-\K^2(x)$ in \eqref{eq:1234}.
\end{proof}
%Hence the FME can be used to analyze outage, when the transmitters form a DPP.
Permanental  point process are the counterparts of DPP which exhibit attraction between the points and are defined by their product densities. The product densities of a permanental point process are defined by the permanent of a kernel matrix, and similar to that of the DPP,  the FME can be used to analyze outage in these processes.

\subsection{Real Data}
If the locations of the nodes in a real wireless system are known, can  a semi-empirical formula of the outage probability be obtained?  The answer is affirmative, and the outage can be obtained using the FME analysis. Given a single snapshot of the node locations, the product densities can be estimated using the techniques described  in \cite{david} and \cite{jolivet}, and the FME can be obtained by numerical integration. Although there is an initial complexity of estimating the product measures,  they can be reused  multiple times to evaluate different functionals of the interference.

\section{Conclusions}
\label{sec:con}
In this paper, we have introduced a new  technique to evaluate Palm averages of functionals of interference in a spatial networks.  A series representation of the interference function, termed as
the factorial moment expansion (FME) was obtained. To obtain this series representation, only the product densities of the underlying transmitter locations are required,  quantities that can be easily computed for many point process. A main contribution of the paper is providing simple sufficient conditions for the FME to hold true for functionals of interference, and computable bounds on the truncation error. We have provided several examples to illustrate this procedure and provided bounds on error. This new technique is versatile and is limited only by the computational complexity and the knowledge about the underlying node distribution.

\appendices

\section{Proof of  Theorem \ref{thm:main}}
For simplicity of notation  we neglect fading, and the case with fading can be easily dealt with since the fading is independent across nodes.
From Hanish's lemma  \cite[Proposition 1]{hanisch1982inversion}, we have 
\begin{equation}
(\P^{(n-1)}_{x_1,\dots,x_{n-1}})^{(1)}_{x_n}(\d \phi)=
\P^{(n)}_{x_1,\dots,x_n}(\d \phi),  \quad \M^{(n)}_\P(\d x_1,\ldots, \d x_{n})\  a.e.
  \label{}
\end{equation}
 Hence,  using the Baccelli-Br\'emaud lemma \cite[Lemma 3.3]{blaszczyszyn1995factorial} applied to the simple point process
 $(\Phi,\P^{(i+1)}_{o,x_1,\dots,x_i})$,
 \begin{align*}
   \int_{M} \F^{(i)}_{x_1,\dots,x_i}(\I(y,\phi))\P^{(i+1)}_{o,x_1,\dots,x_i}(\d
   \phi)
   &=\F^{(i)}_{x_1,\dots,x_i}(0) \\
   &+
\int_{\R^2}\int_M\F^{(i+1)}_{x_1,\dots,x_{i+1}}(\I(y,\phi))\P^{(i+2)}_{o,x_1,\dots,x_{i+1}}(\d
   \phi)\M^{(1)}_{\P^{(i+1)}_{o,x_1,\dots,x_i}}(\d x_{i+1}),
 \end{align*} almost surely with respect to the measure  $\M^{(i)}_{\P_o}(\d x_1,\ldots,\d x_{i})$. 
 Integrating with respect to the measure $\M^{(i)}_{\P_o}$, we obtain
 \begin{align*}
   \int_{\R^{2i}}\int_{M}
&\F^{(i)}_{x_1,\dots,x_{i}}(\I(y,\phi))\P^{(i+1)}_{o,x_1,\dots,x_{i}}(\d   \phi)\M^{(i)}_{\P_o}(\d x_1,\dots \d x_{i}) \\
&=
\int_{\R^{2i}}\F^{(i)}_{x_1,\dots,x_{i}}(0)\M^{(i)}_{\P_o}(\d
x_1,\dots \d x_{i}) \\
   &+
\int_{\R^{2i}}\int_{\R^2}\int_M\F^{(i+1)}_{x_1,\dots,x_{i+1}}(\I(y,\phi))\P^{(i+2)}_{o,x_1,\dots,x_{i+1}}(\d
\phi)\M^{(1)}_{\P^{(i+1)}_{o,x_1,\dots,x_i}}(\d x_{i+1})\M^{(i)}_{\P_o}(\d
x_1,\dots \d x_{i}).
 \end{align*}
From  \cite[Proposition 2.5]{blaszczyszyn1995factorial}, we obtain 
\[\M^{(1)}_{\P^{(i+1)}_{o,x_1,\dots,x_i}}(\d x_{i+1}) \M^{(i)}_{\P_o}(\d
x_1,\dots,\d
x_i) =  \M^{(i+1)}_{\P_o}(\d x_1,\dots,\d
x_{i+1}) ,\]
and hence
 \begin{align*}
   \int_{\R^{2i}}\int_{M}
&\F^{(i)}_{x_1,\dots,x_{i}}(\I(y,\phi))\P^{(i+1)}_{o,x_1,\dots,x_{i}}(\d   \phi)\M^{(i)}_{\P_o}(\d x_1,\dots \d x_{i}) \\
&=
\int_{\R^{2i}}\F^{(i)}_{x_1,\dots,x_{i}}(0)\M^{(i)}_{\P_o}(\d
x_1,\dots \d x_{i}) \\
   &+
\int_{\R^{2(i+1)}}\int_M\F^{(i+1)}_{x_1,\dots,x_{i+1}}(\I(y,\phi))\P^{(i+2)}_{o,x_1,\dots,x_{i+1}}(\d
   \phi) \M^{(i+1)}_{\P_o}(\d x_1,\dots,\d
x_{i+1}).
 \end{align*}
Adding both sides of the above equation for $i=0,\dots,n$, we obtain the
required result.\\

%Applying Theorem 3.1 in \cite{Blasz:97} to the simple point process
%$(\Phi, \PP)$, we
%obtain 
%\begin{align}
%	\EP \F(\I(y,\Phi)) &= \F(o) + \sum_{i=1}^n
%	\int_{\R^{2i}}\F^{(i)}_{x_1,\dots,x_i}(o) \M^{(i)}_{\P_o}(\d
%  x_1,\dots,\d x_i) \nonumber\\
%&+\int_{\R^{2(n+1)}}\int_\mathcal{M}
%\F^{(n+1)}_{x_1,\dots,x_{n+1}}(\I(y,\phi))\P^{(n+2)}_{o,x_1,\dots,x_{n+1}}(\d
%\phi) \M^{(n+1)}_{\P_o}(\d
%x_1,\dots,\d x_{n+1}).
%	\label{eq:asymp}
%  \end{align}
  
\section{Proof of Theorem \ref{thm:inq1}}
\label{sec:proofthm2}
It can be easily seen that the difference function is also equal to
\[\F^{(n)}_{x_1,\dots,x_n}(\I(y,\phi))
=\sum_{\mathcal{P}_n}(-1)^{\sum_{i=1}^n
\b_i}\F\left(\I_{x_n}(y,\Phi)+\sum_{i=1}^n \b_i \h{x_i}\l(x_i
-y)\right),  \]
where $\mathcal{P}_n$ denotes the set of all  binary tuples
$(\b_1,\dots,\b_n)$, $\b_i \in \{0,1\}$ and has a cardinality $2^n$.
For notational convenience we denote $\I_{x_n}(y,\Phi)$ by $\beta$ and
$\h{x_i}\l(x_i
-y)$ by $\gamma_i$. So we have
\[\F^{(n)}_{x_1,\dots,x_n}(\I(y,\phi))
=\sum_{(\b_1,\dots,\b_n)\in \mathcal{P}_n}(-1)^{\sum_{i=1}^n
\b_i}\F\left(\beta+\sum_{i=1}^n \b_i \gamma_i\right).\]
Without loss of generality, we can assume $p_i=i$, $1\leq i\leq k$.
We now partition the set $\mathcal{P}_n$ into $2^{n-k}$ groups. Each
partition
consists of the $n$-binary string with fixed bits in the positions
$\{k+1,\dots,n\}$.
\[A(\b_{k+1},\dots,\b_{n})=\{(\b_1,\dots,\b_k,\b_{k+1},\dots,\b_{n}),
\b_i \in
\{0,1\}, 1\leq i\leq k \}.\] For example with $k=2$ and $n=4$ we
partition of the binary-$4$ tuples as $ A(0,0)$, $A(1,0)$, $A(0,1)$, and
$A(1,1)$, with
\begin{align*}
  A(\b_{3},\b_{4}) = \{(0, 0, \b_3, \b_4 ),(1, 0, \b_3, \b_4 ),(0, 1,
  \b_3, \b_4 ),(1, 1, \b_3, \b_4 )\}.
\end{align*}
So we have
\begin{align}\F^{(n)}_{x_1,\dots,x_n}(\I(y,\phi))
=\sum_{(v_1,\dots,v_{n-k})\in \mathcal{P}_{n-k}}
(-1)^{\sum_{i=1}^{n-k}
\v_i}\H({A(v_1,\dots,v_{n-k})}),
\label{eq:234}
\end{align}
where 
\[\H({A(v_1,\dots,v_{n-k})})=\sum_{(\b_1,\dots,\b_k,v_1,\dots,v_{n-k})\in
A(v_1,\dots,v_{n-k})}(-1)^{\sum_{i=1}^k
\b_i}\F\left(\beta+ \sum_{i=1}^{n-k} v_i z_{i+k}+ \sum_{i=1}^k \b_i
\gamma_i\right). \]
From \eqref{eq:234}, taking the absolute value
\begin{align}
  |\F^{(n)}_{x_1,\dots,x_n}(\I(y,\phi))|
\leq \sum_{(v_1,\dots,v_{n-k})\in \mathcal{P}_{n-k}}
|\H({A(v_1,\dots,v_{n-k})})|,
\label{eq:235}
\end{align}
Define $\F_k(x) = \frac{\d^{k}\F(x)}{\d^{k} x}$, and
\[g(\T_1,\dots,\T_k)= (-1)^k\F_k\left(\beta+\sum_{i=1}^{n-k}v_iz_{i+k}
+\sum_{i=1}^k
(1-\T_i)\gamma_{i}\right)\prod_{i=1}^k \gamma_i.\]
 By a little algebra, it follows that
 \[\int_0^1\dots\int_0^1 g(\T_1,\dots,\T_k)\d \T_1\dots\d \T_k
 = \H({A(v_1,\dots,v_{n-k})}),\]
 and hence 
 \begin{align}
\left| \H({A(v_1,\dots,v_{n-k})})\right| &\leq \int_0^1\dots\int_0^1	|g(\T_1,\dots,\T_k)|\d \T_1\dots\d \T_k,\nonumber\\
	&\leq \FD{k}
\int_0^1\dots\int_0^1 \prod_{i=1}^k \gamma_i\ \d \T_1\dots \d
\T_k,\nonumber\\
	&=\FD{k}  \prod_{i=1}^k
	\gamma_i.
   \label{eq:386}
 \end{align}
 Substituting \eqref{eq:386} in \eqref{eq:235}, 
\begin{align*}
  |\F^{(n)}_{x_1,\dots,x_n}(\I(y,\phi))|
&\leq \sum_{(v_1,\dots,v_{n-k})\in \mathcal{P}_{n-k}}
\FD{k}  \prod_{i=1}^k 	\gamma_i\\
&= 2^{n-k}\FD{k}  \prod_{i=1}^k \gamma_i,
\end{align*}
proving the theorem.
\bibliographystyle{ieeetr}
\bibliography{bib_det}
\end{document}